\documentclass[conference]{IEEEtran}
\IEEEoverridecommandlockouts
\usepackage{makecell}
\usepackage{graphicx}
\usepackage{amsfonts}
\usepackage{amsmath}
\usepackage{algorithm}
\usepackage{subfigure}
\usepackage{amssymb}
\usepackage{algorithmic}
\usepackage{subcaption}
\usepackage{booktabs} 
\usepackage{cite}
\usepackage{amsthm}
\usepackage{nomencl}
\usepackage{extarrows}
\usepackage{hyperref}
\usepackage{cleveref}
\theoremstyle{definition}
\newtheorem{theorem}{Theorem}
\newtheorem{proposition}[theorem]{Proposition}

\theoremstyle{definition}
\newtheorem{definition}{Definition}

\newtheorem{assumption}{\textit{Assumption}}
\def\BibTeX{{\rm B\kern-.05em{\sc i\kern-.025em b}\kern-.08em
    T\kern-.1667em\lower.7ex\hbox{E}\kern-.125emX}}
\begin{document}

\title{\fontsize{19}{22}\selectfont Uncertainty-Aware Capacity Expansion for Real-World DER Deployment via End-to-End Network Integration}

\author{
    \IEEEauthorblockN{Yiyuan Pan*\IEEEauthorrefmark{1}, Yiheng Xie*\IEEEauthorrefmark{2}, Steven Low\IEEEauthorrefmark{2}\thanks{*These authors contributed equally to this work.}}
    \IEEEauthorblockA{
        \IEEEauthorrefmark{1}Shanghai Jiao Tong University, Shanghai, China\\
        Email: pyy030406@sjtu.edu.cn}
    \IEEEauthorblockA{
        \IEEEauthorrefmark{2}California Institute of Technology, Pasadena, CA, USA\\
        Emails: \{yxie5, slow\}@caltech.edu}
}

\maketitle

\begin{abstract}
The deployment of distributed energy resource (DER) devices plays a critical role in distribution grids, offering multiple value streams, including decarbonization, provision of ancillary services, non-wire alternatives, and enhanced grid flexibility. However, existing research on capacity expansion suffers from two major limitations that undermine the realistic accuracy of the proposed models: (i) the lack of modeling of three-phase unbalanced AC distribution networks, and (ii) the absence of explicit treatment of model uncertainty. To address these challenges, we develop a two-stage robust optimization model that incorporates a 3-phase unbalanced power flow model for solving the capacity expansion problem. Furthermore, we integrate a predictive neural network with the optimization model in an end-to-end training framework to handle uncertain variables with provable guarantees. Finally, we validate the proposed framework using real-world power grid data collected from our partner distribution system operators. The experimental results demonstrate that our hybrid framework, which combines the strengths of optimization models and neural networks, provides tractable decision-making support for DER deployments in real-world scenarios. 
\end{abstract}

\section{Introduction}
\subsection{Background}
In recent years, to meet the growing demands for power supply reliability and decarbonization, substantial planning and construction of new transmission and energy storage facilities have been undertaken. This process is commonly referred to as capacity expansion. Among the new additions, a particular class of hardware known as distributed energy resource (DER) devices—such as solar panels, battery systems, and electric vehicle charging stations—has garnered significant attention, as they play a pivotal role in integrating local renewable energy and enhancing overall energy dispatch flexibility. The widespread deployment of DER devices has brought profound changes to net consumption behaviors, which are characterized by greater uncertainty in electricity usage patterns \cite{13}. For example, DER devices such as distributed energy storage systems can enable power utility to smooth regional load curves to enable peak shaving through demand response. Consequently, the study of capacity expansion and subsequent optimal dispatch has emerged as a promising research topic. Modeling and analyzing the impacts of DER deployment on future electricity consumption behaviors and grid performance can provide valuable insights for power utilities.

\subsection{Related Works}
\vspace{-0.1cm}
Existing research on capacity expansion often oversimplifies the modeling of distribution networks by treating three-phase unbalanced systems as single-phase equivalents, overlooking the critical complexities and operational challenges inherent in realistic networks. Several studies \cite{1,2,3,4} have formulated the capacity expansion and optimal dispatch problem as a sequential two-stage optimization framework. In the first stage, the optimal investment strategy for renewable energy generation is determined. In the second stage, the optimal power dispatch strategy is formulated for the transmission network to minimize operational costs. Building upon this foundation, \cite{5} further incorporates electrical and topological constraints at the distribution side level. Unlike earlier works, this study extends the second stage decision variables to include bus voltages and voltage angles, beyond merely treating power flows as decision variables. While these studies provide valuable insights into the sequential and multi-stage nature of the capacity expansion problem, they fail to capture the impact of realistic regional distribution side structures, including their electrical and topological characteristics, on dispatch strategies. Studies \cite{1,2,3,4} lack detailed modeling of the distribution side and remain confined to the transmission level, whereas \cite{5} limits its analysis to single-phase AC networks, which is inconsistent with practical three-phase AC distribution systems. 

Even with the development of three-phase AC distribution side models, they alone cannot serve as effective tools for practical decision-making due to the increasing uncertainty in real-world systems, such as real-time electricity prices and load demands. To address these challenges, \cite{7,8} employ predefined Gaussian probability distributions to represent load uncertainties. Additionally, \cite{9} adopts a general polyhedral set to describe these uncertainties and solves the multi-stage problem with uncertainties using the column-and-constraint generation (C\&CG) algorithm. Furthermore, \cite{10} refines the modeling approach by representing uncertainties as deviations between potential real values and historical data, parameterized by predefined distributions, resulting in more accurate solutions. While these works enhance the realism of capacity expansion modeling and improve the effectiveness of the solutions by introducing uncertainty variables, they lack the ability to adaptively predict these uncertainties based on the specific operating environment of the regional grid. This limitation implies that in dynamically changing real-world conditions, such as weather variations or regional characteristics that cause shifts in the underlying data distributions, these models lose their validity and effectiveness.

\subsection{Methodology}
In summary, existing research on capacity expansion and optimal dispatch faces two key challenges in model fidelity. (i) Current capacity expansion studies lack the establishment of distribution side models from the perspective of three-phase AC systems, thereby neglecting the impacts of electrical structures and grid topologies on dispatch strategies. (ii) Conventional approaches to modeling uncertainty in optimal dispatch fail to incorporate predictive capabilities, which compromises the adaptability of such models.

To address these issues, this paper proposes a robust two-stage optimization model for capacity expansion with DER deployment and subsequent dispatch, integrated with a hybrid framework that end-to-end trains a predictive neural network alongside the optimization model. First, in the proposed two-stage robust optimization model, the inner-stage optimal dispatch problem is formulated as a maximization-minimization problem based on the linear DistFlow model, explicitly considering the presence of transformers in the real-world grid.

To address the challenges posed by introducing uncertain variables, which transform the problem into a semi-infinite optimization problem that is computationally intractable, we leverage the properties of affine policies and uncertainty sets. By doing so, we transform the two-stage robust optimization model into a single-stage optimization problem. This transformation not only simplifies the solution process but also enables gradient computation, which is essential for integrating the optimization model with neural network training in an end-to-end framework. Then, by leveraging the properties of affine policies and uncertainty sets, we transform the two-stage robust optimization model into a single-stage optimization problem, as two-stage problems are typically solved iteratively and do not permit gradient computation. By doing so, we are also able to overcome the computational challenges arising from the semi-infinite nature introduced by uncertain variables. Finally, we utilize a hybrid framework that integrates predictive neural networks with the downstream single-stage optimization model, enabling accurate prediction of uncertainties while simultaneously determining the optimal dispatch strategy (see Fig \ref{fig1}). Inspired by \cite{12,999}, to accelerate the convergence of the predictive network for task-specific learning and enhance its performance in the capacity expansion task, both the prediction loss and the optimization objective function are utilized to update the neural network. To validate the effectiveness of the proposed framework, we employ real-world data collected from a regional grid in Southern California. Experimental results demonstrate that the framework not only achieves accurate uncertainty predictions but also optimizes subsequent dispatch strategies based on the predictions.

The key contributions of this paper are as follows:
\begin{itemize}
    \item We develop a robust two-stage optimization model for capacity expansion with DER device deployment, explicitly designed for three-phase AC distribution sides. Specifically, the model accounts for the grid topology and the injection behavior at the device level for each bus, enabling it to capture the characteristics of real power systems more accurately compared to existing models.
    \item Our task-based hybrid framework can make adaptive and real-time predictions concerning each input and provide corresponding scheduling strategies. In particular, it predicts real-time electricity prices and loads based on inputs such as local weather conditions and integrates these predictions into subsequent optimization processes. This approach significantly enhances the adaptability of the overall framework.
    \item We collect and utilize real-world grid data from Southern California to validate the proposed method. To the best of our knowledge, this is the first study to employ actual three-phase AC distribution data for validation. The results enable a comprehensive evaluation of the model's efficacy and reliability.
\end{itemize}

\section{Capacity Expansion Optimization Problem}
\subsection{Problem Overview}
A typical distribution network comprises a set of buses and the lines connecting these buses. Without loss of generality, let $\mathcal{N}=\{0,1,...,N\}$ denote the set of buses, and define $\mathcal{N}^+=\mathcal{N}\setminus\{0\}$. Each line connects a directed pair $(i,j)$ of buses. Let $\mathcal{E}$ denote the set of lines, where $(i,j)\in\mathcal{E}$. Consistent with most existing works, the dispatch process is formulated as a two-stage robust optimization model. As a power utility, we need to determine whether to deploy DER devices at these $N$ locations and make subsequent dispatch decisions. The optimization is performed over a planning horizon of $\mathcal{T}=\{0,1,...,\tau\}$. In our paper, DERs include battery energy storage systems (BESS) and photovoltaic (PV) panels. Each bus is assumed to have PV panels already installed, and the decision to be made is whether to deploy BESS. The objective is to minimize operational costs while ensuring the satisfaction of grid constraints.

\begin{figure}
    \centering
    \includegraphics[width=0.47\textwidth]{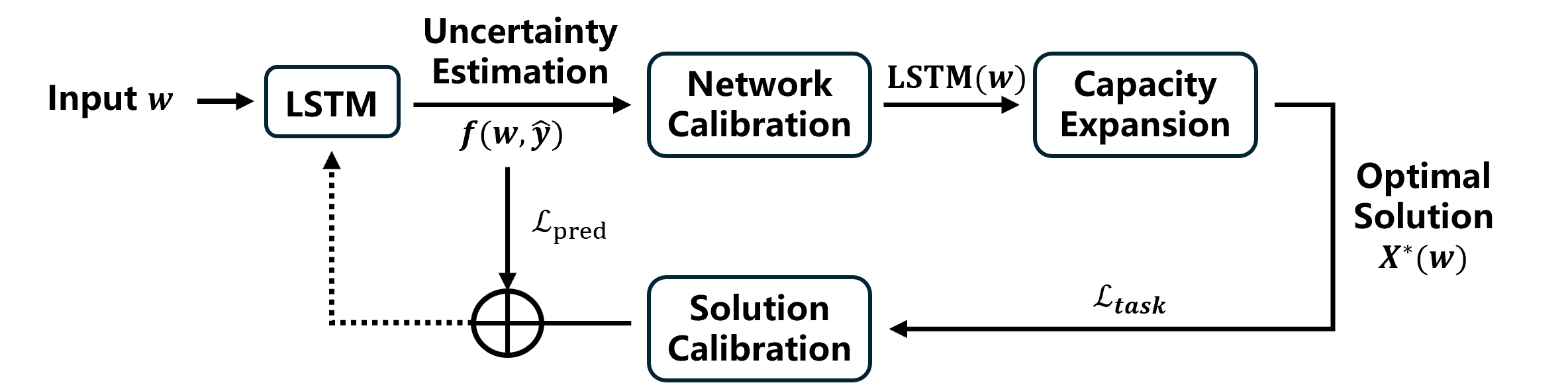}
    \caption{Hybrid training framework for capacity expansion.}
    \label{fig1}
\end{figure}

\subsection{Objective Function}
In our two-stage robust optimization model, the objective function aims to minimize the deployment costs of DER devices while achieving the minimization of grid power procurement costs under the worst-case scenario. The minimization of procurement costs corresponds to the optimal power flow (OPF) problem. For the distribution network model in the OPF problem, we adopt the branch flow model (BFM) from \cite{11} under the assumption that it satisfies Assumption \ref{assu1} outlined below.
\begin{assumption}\label{assu1}
The distribution network satisfies the following conditions:
\begin{enumerate}
    \item The network topology is radial (tree network).
    \item The admittance matrix of the line between any two buses is invertible (this assumption will be relaxed in Section \ref{real}).
    \item Line losses during power transmission are negligible.
    \item The voltages at each bus are balanced, satisfying a three-phase system with a $120^\circ$ phase difference between phases.
\end{enumerate}
\end{assumption}

Under this assumption, the distribution network's BFM simplifies to the linear DistFlow model \cite{11}. Furthermore, due to the radial structure and power conservation, the power procurement cost can be directly calculated based on the real power at the slack bus (bus $0$). Therefore, the objective function is expressed as:
\begin{gather}
\min_{x}c^Tx+\max_{\hat{y_P},\hat{y_D}}\min_{z}\sum_{t=0}^{\tau}\hat{y_P^t}^T\text{Re}(s_0^t)
\end{gather}
where $x\in\mathbb{R}^{N}$ represent the decision variable for the sizing and siting of BESS during the capacity expansion process, where an entry $x_i = 0$ indicates no BESS installation at bus $i$. The parameter $c\in\mathbb{R}^{N}$ represents the installation cost of BESS at each bus. In the inner max-min problem, $\hat{y}_p\in\mathbb{R}^{3}$ and $\hat{y}_d\in\mathbb{R}^{3}$ are the uncertain variables representing the real-time electricity prices and load demands, respectively. The operator $\text{Re}(\cdot)$ denotes the real part of a complex value. For the BFM, $z = \{s,v\}$, where $s_j^t\in\mathbb{C}^{3}$ denoted the complex power at each bus and $v_j^t\in\mathbb{C}^{3\times3}$ represents the squared magnitude matrices of bus voltage, for each $j\in\mathcal{N}, t\in\mathcal{T}$. The variable $s_0$ denotes the complex power at the slack bus.

\subsection{OPF Constraints}
For the inner max-min OPF problem, we first define the following power flow constraints:
\begin{subequations}
\begin{align}
    & \begin{bmatrix}
        v_j & S_{jk} \\
        S_{jk}^H & i_{jk}
    \end{bmatrix} \ge 0, 
    \forall j\rightarrow k\in \mathcal{E} \label{2a}\\
    & \text{rank}\left(\begin{bmatrix}
        v_j & S_{jk} \\
        S_{jk}^H & i_{jk}
    \end{bmatrix}\right)=1,
    \forall j\rightarrow k\in \mathcal{E} \label{2b}
\end{align}
\end{subequations}

Although the BFM dose not include decision variables $V_j$ or $i_{jk}$, the positive semi-definite rank-1 constraints (\ref{2a})(\ref{2b}) ensures that there exist $V_j$ and $I_{jk}$ such that
\begin{equation}
    v_j=V_jV_j^H, \ i_{jk}=I_{jk}I_{jk}^H, S_{jk}=V_jI_{jk}^H, \forall j\rightarrow k\in\mathcal{E} \label{3}\\
\end{equation}

Next, we present the equality constraints derived from the network topology within the three-phase linear DistFlow model:
\begin{subequations}
\begin{align}
    & \sum\nolimits_{j\in\mathcal{N}}s_j=0 \label{4a} \\
    & \lambda_{ij}=-\sum\nolimits_{k\in T_j}s_k, \forall j\rightarrow k\in\mathcal{E} \label{4b} \\
    & S_{ij}=\gamma\cdot\text{diag}(\lambda_{ij}), \forall j\rightarrow k\in\mathcal{E} \label{4c} \\
    & v_0-v_j=\sum\nolimits_{(j,k)\in P_j}(z_{jk}S_{jk}^H+S_{jk}z_{jk}^H), \forall j\in\mathcal{N} \label{4d}
\end{align}
\end{subequations}

Constraint (\ref{4a}) enforces power balance across the network, ensuring that the total injected power at all buses sums to zero. Constraint (\ref{4b}) defines the auxiliary variable $\lambda_{ij}$, where $T_j$ is the subtree rooted at bus $j$, including $j$. Constraint (\ref{4c}) relates the complex power matrix $S_{ij}$ to the auxiliary variable $\lambda_{ij}$ through a constant phase vector $\gamma$ and the diagonal operator. Constraint (\ref{4d}) calculates the voltage difference between the slack bus and bus $j$ as a function of power flow $S_{jk}$ and line impedance $z_{jk}$ along the unique path $P_j$ from bus $0$ to bus $j$, where $(\cdot)^{H}$ represents the Hermitian transpose of the matrix. The detailed derivation and theoretical background of these constraints can be found in \cite{11}. The definition of $\gamma$ is given as follows:
\begin{equation}
    \gamma = \begin{bmatrix}
        1 & \alpha^2 & \alpha \\
        \alpha & 1 & \alpha^2 \\
        \alpha^2 & \alpha & 1
    \end{bmatrix}, \ \alpha:= \mathbf{e}^{-\mathbf{i}2\pi/3} \label{5}
\end{equation}

We now present the general operational constraints for power system management.
\begin{subequations}
\begin{align}
    & s_j^{\text{min}}\le s_j \le s_j^{\text{max}}, \forall j\in \mathcal{N} \label{6a} \\
    & v_j^{\text{min}}\le \text{diag}(v_j) \le v_j^{\text{max}}, \forall j\in \mathcal{N} \label{6b} \\
    & \text{diag}(i_{jk})\le i_{jk}^{\text{max}} \label{6c}
\end{align}
\end{subequations}

Finally, we present the constraints related to the deployment and operation of DER devices at each bus.
\begin{subequations}
\begin{align}
    & \text{Re}(s_j)=\text{Re}(s_j^d)+\text{Re}(s_j^p)+\hat{y_D}, \forall j\in \mathcal{N} \label{7a} \\
    & s^{d,\text{min}}\le s_j^{d,t} \le s^{d,\text{max}}, \forall j\in \mathcal{N} \label{7b} \\
    & \text{SOC}_{j}^{t+1} = \text{SOC}_{j}^{t}+\text{Re}(s_j^{d,t}), \forall j\in\mathcal{N},t\in\mathcal{T}\setminus\{\tau\} \label{7c} \\
    & \text{SOC}_{j}^{0} = \mathbf{0}, \forall j\in \mathcal{N} \label{7d} \\
    & \text{SOC}^{\text{min}}\cdot (x)_j \le \text{SOC}^t_{j} \le \text{SOC}^{\text{max}}\cdot (x)_j, \forall j\in \mathcal{N} \label{7e}
\end{align}
\end{subequations}

Constraint (\ref{7a}) defines the components of power injection at each bus, where $s_j^d$ represents the BESS dispatch power at bus $j$, $s_j^p$ denotes the power generated by the PV panels at bus $j$, and $y_d$ is the load demand at the same bus. The PV power $s_j^p$ is assumed to be constant at each time step and is derived from historical data. Constraint (\ref{7b}) sets the operational bounds for BESS dispatch, ensuring that the dispatched power remains within the allowable range. Constraints (\ref{7c}) and (\ref{7d}) specify the state-of-charge (SOC) dynamics of each BESS, including the recursive formula for SOC evolution from time $t$ to $t+1$ and the initial SOC values. Moreover, Constraint (\ref{7d}) indicates that the BESS capacity is directly tied to the power utility's capacity expansion decisions $x$, where $(\cdot)_j$ denotes its $j$-th entry.

We now formulate the capacity expansion problem as a two-stage model, which is expressed as follows:
\begin{align}
    & \min_{x}c^Tx+\max_{\hat{y_P},\hat{y_D}}\min_{z}\sum_{t=0}^{\tau}\hat{y_P^t}^T\text{Re}(s_0^t)  \\
    & \ s.t. \ (2), (4)-(7) \notag
\end{align}

\subsection{Non-invertible Admittance Matrices}\label{real}
However, Constraint (\ref{4d}) in the above model is not realistic. In real-world distribution networks, $\Delta Y$-configured transformers are widely present. Due to the existence of such transformers, the admittance matrix between bus $i$ and bus $j$ is generally asymmetric and non-invertible, which implies that the second condition in Assumption \ref{assu1} is not valid. As a result, the constraints between $j\rightarrow i$ and $i\rightarrow j$ should have distinct forms and should not be expressed directly in terms of impedance. In doing so, we removed the second condition from Assumption \ref{assu1}.

To reformulate this constraint, we begin by considering the three-phase circuit, where series admittances $y_{jk}^s \neq y_{kj}^s$ and shunt admittance $y_{jk}^m = \mathbf{0}$ under Assumption \ref{assu1}. Based on Kirchhoff's Current Law (KCL) and Ohm's Law, the following equations can be derived.
\begin{align}
    & \begin{cases}
    I_{jk} = y_{jk}(V_j - V_k), S_{jk}=V_jI_{jk}^H, j\rightarrow k \in \mathcal{E} \\
    I_{kj} = y_{kj}(V_k - V_j), S_{kj}=V_kI_{kj}^H, k\rightarrow j \in \mathcal{E}
    \end{cases}
\end{align}

Taking $j\rightarrow k$ as an example, we substitute the expression for current $I$ into the complex power equation. 
\begin{equation}
    S_{jk}^Hy_{jk}^H=y_{jk}v_jy_{jk}^H-y_{jk}(V_kV_j^H)y^H_{jk}, j\rightarrow k \in \mathcal{E}
\end{equation}

Next, under the negligible line loss assumption ($i_{jk} = \mathbf{0}$), we rewrite voltage $V_j=V_j-V_k+V_k$ and substitute it into the above equations.
\begin{equation}
    S^H_{jk}y^H_{jk}+[y_{jk}V_kI^H_{jk}+i_{jk}]=y_{jk}(v_j-v_k)y_{jk}^H, j\rightarrow k \in \mathcal{E}
\end{equation}
where
\begin{equation}
    y_{jk}V_kI_{jk}^H+i_{jk}=[y_{jk}V_k+y_{jk}(V_j-V_k)]I^H_{jk}=y_{jk}S_{jk}
\end{equation}

Finally, we obtain:
\begin{align}
    & \begin{cases}
    y_{jk}(v_j-v_k)y_{jk}^H=S^H_{jk}y_{jk}^H+y_{jk}S_{jk}, j\rightarrow k \in \mathcal{E} \\
    y_{kj}(v_k-v_j)y_{kj}^H=S^H_{kj}y_{kj}^H+y_{kj}S_{kj}, k\rightarrow j \in \mathcal{E}
    \end{cases} \label{13}
\end{align}

However, the derived Constraint (\ref{13}) cannot be directly used to solve for each element of the $S$ matrix due to the presence of additional degrees of freedom. The derivation process of Constraint (\ref{4d}) in \cite{11} employs the balanced voltage condition from Assumption \ref{assu1}, which eliminates all degrees of freedom and enables the computation of every element in the $S$ matrix. Therefore, additional constraints must be imposed on $S$ or $v$. Moreover, in our OPF problem, the constraints on $v_j$, as expressed in Constraint \ref{3} or its equivalent form in Constraints (\ref{2a})(\ref{2b}), correspond to rank-1 constraints. These are intractable for existing solvers in computational environments. To address these two issues, we need to incorporate the balanced voltage assumption to reformulate and derive new constraints for $v$.

Under the balanced voltage assumption, the voltage at bus $j$ can be represented as:
\begin{equation}
    V_j=V_j^a\cdot\begin{bmatrix}
        1 \\
        \alpha \\
        \alpha^2
    \end{bmatrix} = V_j^a\alpha_+, \forall j\in \mathcal{N},V_j^a\in \mathbb{C}
\end{equation}
where $(\cdot)^a$ denotes the first phase (phase $a$) of the corresponding three-phase system $(a, b, c)$. Consequently, the voltage matrix $v_j$ can be reformulated as:
\begin{equation}
    v_j=V_jV_j^H=V_j^a{V_j^a}^*\alpha_+\alpha_+^H=v_j^c A, \forall j\in \mathcal{N}
\end{equation}
where $(\cdot)^*$ denotes the conjugate operation, $v_j^c$ is a scalar real number, and $A$ is a known rank-1 matrix. This transformation does not introduce any new decision variables.

To this end, we derive a solvable and realistic model for the capacity expansion problem.
We now formulate the capacity expansion problem as a two-stage model, which is expressed as follows:
\begin{align}
    & \min_{x}c^Tx+\max_{\hat{y_P},\hat{y_D}}\min_{z}\sum_{t=0}^{\tau}\hat{y_P^t}^T\text{Re}(s_0^t)  \label{16}\\
    & \ s.t. \ (4a)-(4c), (5)-(7), (13), (15) \notag
\end{align}

\subsection{Single-stage Reformulation}
However, due to the iterative nature of solving two-stage problems, the feasibility of the solution cannot be guaranteed, and the obtained solution may not be exact. This limitation prevents the direct use of the problem formulation in (\ref{16}) for end-to-end training and gradient update of the network. Therefore, it is necessary to reformulate the two-stage robust optimization problem into a single-stage optimization layer. Specifically, this requires eliminating the uncertainty layer and transforming the original min-max-min problem into a standard minimization problem, enabling the computation of gradient updates for network training. For simplicity, the uncertainties considered in this paper are represented using the box uncertainty set defined as follows:
\begin{equation}
    \hat{y_P} \in [\underline{y_P},\overline{y_P}]=\mathcal{Y}_P, \ \hat{y_D} \in [\underline{y_D},\overline{y_D}]=\mathcal{Y}_D
\end{equation}

\subsubsection{Electricity Price Uncertainty}
Given that the problem is formulated as a robust optimization problem, which focuses solely on the worst-case scenario, and that electricity prices are positively correlated with the objective function, they can be directly set to their worst-case values, i.e., their maximum values $\overline{y_P}$.

\subsubsection{Load Demand Uncertainty}
For the uncertain variable $y_d$, its coupling with the constraints leads to the optimization problem being classified as a semi-infinite problem. In this case, inspired by \cite{12}, we employ an affine policy to reformulate the problem. The core idea is to identify the dominant set of the uncertain variable, enabling us to approximate the solution of the original two-stage robust optimization problem by solving a simplified single-level problem over the dominant set. To ensure the completeness of our work, we provide the necessary details and proofs for certain propositions. 

For the box uncertainty set of load demand, the following Assumption \ref{assu2} holds.
\begin{assumption}\label{assu2}
    An uncertainty set $\mathcal{U}\subseteq[0,1]^m$ is convex, full-dimensional with $\mathbf{e}_i\in \mathcal{U}$ for all $i=1,...,m$ and down-monotone, i.e., $h\in\mathcal{U}$ and $0\le h'\le h$ implies that $h'\in\mathcal{U}$
\end{assumption}

We then give the definition of a dominant set concerning the uncertainty set.
\begin{definition}
    (\textit{Dominant Set}) Given an uncertainty set $\mathcal{U}\subseteq\mathbb{R}^M_+$, $\hat{\mathcal{U}}\subseteq\mathbb{R}^M_+$ dominates $\mathcal{U}$ if for all $h\in\mathcal{U}$, there exists $\hat{h}\in\hat{\mathcal{U}}$ such that $\hat{h}\ge h$.
\end{definition}
For the box uncertainty set, we present the following proposition:
\begin{proposition}
    For a box uncertainty set $\mathcal{U}=\{h\in[0,1]^m$ $|\sum\nolimits_{i=1}^m h_i\le k\}$, we have $\hat{\mathcal{U}}=\beta\cdot\text{conv}(\mathbf{e}_1,...,\frac{k}{m}\mathbf{e})$ as its dominant set, where $\beta=\text{min}(k,\frac{m}{k})$.
\end{proposition}
\begin{proof}
    To prove that $\hat{\mathcal{U}}$ is the dominant set, it suffices to consider $h$ on the boundary of $\mathcal{U}$, i.e., satisfying $\sum\nolimits_{i=1}^{m+1} h_i = k$, and to find $(\alpha_1, \alpha_2, \ldots, \alpha_{m+1})$ as non-negative real numbers such that $\sum\nolimits_{i=1}^{m+1} \alpha_i = 1$ and, for all $i \in [m]$, the inequality $h_i \leq \beta \left( \alpha_i + \frac{k}{m} \alpha_{m+1} \right)$ holds.
    
    In the first case, when $\beta = k$, we choose $\alpha_i = \frac{h_i}{k}$ for $i \in [m]$ and $\alpha_{m+1} = 0$. Clearly, $\sum\nolimits_{i=1}^{m+1} \alpha_i = 1$ holds, and for all $i \in [m]$, we have:
    \[
    \beta \left( \alpha_i + \frac{k}{m} \alpha_{m+1} \right) = k \cdot \frac{h_i}{k} = h_i.
    \]
    
    In the second case, when $\beta = \frac{m}{k}$, we choose $\alpha_i = 0$ for $i \in [m]$ and $\alpha_{m+1} = 1$. Again, $\sum\nolimits_{i=1}^{m+1} \alpha_i = 1$ holds, and for all $i \in [m]$, we have:
    \[
    \beta \left( \alpha_i + \frac{k}{m} \alpha_{m+1} \right) = \frac{m}{k} \cdot \frac{k}{m} = 1 \geq h_i.
    \]    
\end{proof}

Thus, in our problem, the uncertain variable $\hat{y_D} \in [\underline{y_D}, \overline{y_D}]$ can be reformulated as $\hat{y_D} \in \mathcal{Y}_D=\frac{\underline{y_D} + \overline{y_D}}{2} + \frac{\overline{y_D} - \underline{y_D}}{2} \cdot [\mathbf{0}, \mathbf{1}]$. Based on this reformulation, we identify that $m = 1$ and $k = 1$ for the actual uncertain component. Under these circumstances, the value of $\beta$ is set to 1. Therefore, for the uncertain set $\mathcal{Y}_D$, its dominant set $\hat{\mathcal{Y}_D}$ can be determined by taking the maximum value for each dimension.
\begin{equation}
    \hat{\mathcal{Y}_D}=\overline{y_D} = \left\{ y_D' \ \middle| \ (y_D')_i = \max_{\hat{y_D} \in \mathcal{Y}_D} (\hat{y_D})_i, \ \forall i \in [3] \right\}
\end{equation}

After obtaining the dominant set, we present the proposition from \cite{12}, which defines both the original problem and the approximate problem while providing the approximation accuracy between the two.
\begin{proposition}
For a general two-stage robust optimization problem:
\begin{align}
    & z(\mathcal{U})=\min c^Tx+\max_{h\in\mathcal{U}}\min_{y(h)}d^Ty(h) \notag
\end{align}
\begin{align}
    & \ s.t. \ Ax+By(h)\ge h, \ \forall h\in \mathcal{U} \notag \\
    & \qquad \  x \in \mathbb{R}_+, \ y(h) \in \mathbb{R}_+ \notag
\end{align}
where $\mathcal{U}$ is the uncertainty set. Using the dominant set $\hat{\mathcal{U}}$, we derive the approximation problem
\begin{align}
    & z(\hat{\mathcal{U}}) = \min_{x, z} \ c^T x + z, \notag \\
    & s.t. \ z \geq d^T y_i, \ \forall i \in [m+1] \notag \\
    & \qquad Ax + B y_i \geq \beta e_i, \ \forall i \in [m] \notag \\
    & \qquad Ax + B y_{m+1} \geq \beta v \notag \\
    & \qquad x \in \mathbb{R}_+, \ y_i \in \mathbb{R}_+, \ \forall i \in [m+1] \notag
\end{align}
and corresponding approximation bound:
\[z(\hat{\mathcal{U}}) \leq z(\mathcal{U}) \leq \beta \cdot z(\hat{\mathcal{U}})\]
\end{proposition}
\begin{proof}
Let $(\hat{x}, \hat{y}(\hat{h}), \hat{h} \in \hat{\mathcal{U}})$ be an optimal solution for $z(\hat{\mathcal{U}})$. For each $h \in \mathcal{U}$, define $\tilde{y}(h) = \hat{y}(\hat{h})$. Therefore, the following holds:
\[
A \hat{x} + B \tilde{y}(h) = A \hat{x} + B \hat{y}(\hat{h}) \geq h,
\]
which implies $z(\mathcal{U}) \leq z(\hat{\mathcal{U}})$.

Conversely, let $(x^*, y^*(h), h \in \mathcal{U})$ be a feasible solution for $z(\mathcal{U})$. Then, for any $\hat{h} \in \hat{\mathcal{U}}$, since $\frac{\hat{h}}{\beta} \in \mathcal{U}$, we have:
\[
A x^* + B y^*\left(\frac{\hat{h}}{\beta}\right) \geq \frac{\hat{h}}{\beta}.
\]

By scaling, this implies:
\[
A (\beta x^*) + B (\beta y^*)\left(\frac{\hat{h}}{\beta}\right) \geq \hat{h}.
\]

Thus, the feasible solution $\left(\beta x^*, \beta y^*\left(\frac{\hat{h}}{\beta}\right), \hat{h} \in \hat{\mathcal{U}}\right)$ is also feasible for $z(\hat{\mathcal{U}})$, and we have: $z(\hat{\mathcal{U}}) \leq \beta \cdot z(\mathcal{U})$.
\end{proof}
To this end, we can rewrite Constraint (\ref{7b}) and derive the single-layer formulation of the original two-stage problem.
\begin{align}
    & \min_{x,s,v}c^Tx+\eta  \label{18}\\
    & \ s.t. \ \eta\ge\overline{y_P}\cdot\text{Re}(s_0) \notag \\
    & \qquad  \text{Re}(s_j)=\text{Re}(s_j^d)+\text{Re}(s_j^p)+y_D', \forall j\in \mathcal{N}, y_D'\in \hat{\mathcal{Y}_D} \notag \\
    & \qquad  (4a)-(4c), (5)-(6), (7b)-(7e), (13), (15) \notag
\end{align}

\section{Hybrid Traning Framework}
To enhance decision-making robustness and improve performance in the capacity expansion task, we employ a hybrid framework that integrates the optimization model into the training process of the predictive network. Specifically, the predictive network and the optimization model form an end-to-end closed-loop system: the network provides predictions of uncertainties, which serve as parameters for the optimization model, while the optimization model leverages its objective function to update the network parameters, thereby improving the network's task-specific performance.
\subsection{Network Module}
We employ an LSTM-based network, which can handle inputs and outputs of arbitrary lengths, to predict future electricity prices and expected loads. The input to the network, denoted as $\mathbf{w}$, consists of $M$ historical time steps of real-time electricity prices, weather conditions (wind speed, temperature, humidity, and solar radiation intensity), and the actual loads of the $N$ buses. In total, the input includes ($N \cdot M + 4 \cdot M + 1$) features. The network predicts the expected loads and real-time electricity prices for the next $\tau+1$ time steps, where $\tau+1$ corresponds to the optimization horizon of the model. 

To enhance robust decision-making, the network outputs uncertainty estimates for these variables, represented as box uncertainty in this paper, denoted as $\mathcal{Y}_D$ and $\mathcal{Y}_P$. For the obtained box uncertainty, we compute the loss using the negative log-likelihood (NLL) between the predicted uncertainty bounds and the ground truth values $Y_P=\{y_P^0,...,y_P^\tau\}$ and $Y_D=\{y_D^0,...,y_D^\tau\}$ over the $\tau$ time steps. This loss is used to update the network parameters.
\begin{equation}
    \mathcal{L}_{pred}=\sum_{t=0}^\tau \text{NLL}(\underline{y_P}^t,\overline{y_P}^t,y_P^t)+\text{NLL}(\underline{y_D}^t,\overline{y_D}^t,y_D^t)
\end{equation}

However, a more precise metric is required to evaluate the coverage level of the predicted box uncertainty relative to the ground truth, which is essential for estimating the effectiveness of decisions derived from the optimization model. Let $\mathcal{Y}=\mathcal{Y}_P\cup\mathcal{Y}_D$ and $\hat{y}\in\mathcal{Y}$. Inspired by \cite{12}, we represent the process of generating box uncertainty from the LSTM network as the following general form:
\begin{equation}
    \text{LSTM}(w)=\{\hat{y}\in\mathbb{R}^3 | s(w,\hat{y})\leq q\}
\end{equation}
where $s:\mathbb{R}^{NM+4M+1}\times\mathbb{R}^3\rightarrow \mathbb{R}$ is a nonconformity score function and $q\in\mathbb{R}$. To select $q$ during inference, we apply a split conformal prediction process, ensuring that $\text{LSTM}(w)$ provides a marginal coverage at a certain coverage level $\alpha$, as formalized in the following Definition \ref{def2}.
\begin{definition}\label{def2}
    (\textit{marginal coverage}) An uncertainty set for an unknown data distribution $\mathcal{P}$ provides marginal coverage at the $(1-\alpha)$-level if $\mathbb{P}_{(w,\hat{y})\sim\mathcal{P}}(\hat{y}\in\text{LSTM}(w))\in 1-\alpha$.
\end{definition}

Then, we use Proposition \ref{prop3} from \cite{15} to represent the relationships between $q$ and $\alpha$ (see Appendix D in \cite{15} for the proof).
\begin{proposition}\label{prop3}
Let the dataset $\mathcal{D} = \{(w_n, \hat{y}_n)\}_{n=1}^M$ be sampled i.i.d. from the implicit distribution $\mathcal{P}$ obtained during the training phase. Let $q$ be the $(1 - \alpha)$-quantile for the set $\{s(w_n, \hat{y}_n)\}_{n=1}^M$, then $\text{LSTM}(w)$ satisfies the following guarantee:
\[
1 - \alpha \leq \mathbb{P}_{w, \hat{y}, \mathcal{D}}(\hat{y} \in \text{LSTM}(w)) \leq 1 - \alpha + \frac{1}{M + 1}.
\]
\end{proposition}

Additionally, For the box uncertainty adopted in our paper, the score function $s$ can be represented as follows:
\begin{equation}
    s_\theta(w,\hat{y})=\text{max}\cup^m_{i=1}\left\{f_\theta^{\frac{\alpha}{2}}(w)_i-\hat{y}_i, \hat{y}_i-f_\theta^{1-\frac{\alpha}{2}}(w)_i\right\}
\end{equation}
where $f_\theta^a(w)$ is the predictive model to estimate the $a$-quantile and $\theta$ are the model weights. 

Therefore, we can compute $q$ following Proposition \ref{prop3} to satisfy the coverage requirement. This is achieved by sorting the batch of input data $\{s(w_n, \hat{y}_n)\}_{n=1}^M$ in ascending order by scores and selecting the $\lfloor(M+1)(1-\alpha)\rfloor$-th element of the sort as the threshold $q^*$. $\lfloor \cdot \rfloor$ represents the floor operation.

With the chosen $q^*$, the uncertainty set $\mathcal{Y}$ becomes:
\begin{multline}
    \mathcal{Y}= \mathcal{Y}_P\cup\mathcal{Y}_D =\{\hat{y}\in\mathbb{R}^3 | s(w,\hat{y})\leq q^*\}\\
    :=\left[f_\theta^{\frac{\alpha}{2}}(w)-q^*\mathbf{1}, f_\theta^{1-\frac{\alpha}{2}}(w)+q^*\mathbf{1}\right]
\end{multline}

\subsection{Optimization as a Differentiable Module}
Our end-to-end framework takes the optimization objective function of Model (\ref{18}) as the task loss. This requires the optimization problem to be end-to-end differentiable to enable gradient back-propagation. We denote the optimal value derived from solving Model (\ref{18}) as $\mathcal{L}_{task}$ and corresponding decision variables as $X^*$$=$$(x,s,v)^*$. Then, the gradient update value can be computed using the chain rule, with $\frac{\partial X*}{\partial \theta}$ computed via differentiating through Karush–Kuhn–Tucker (KKT) conditions.
\begin{equation}
    \frac{\partial \mathcal{L}_{task}}{\partial \theta}=\frac{\partial \mathcal{L}_{task}}{\partial X^*} \frac{\partial X^*}{\partial \theta}
\end{equation}

Ultimately, we derive two loss functions: one to improve the overall predictive capability of the framework and another to enhance its task-specific performance.
\begin{equation}
    \mathcal{L}=\lambda\cdot\mathcal{L}_{pred}+(1-\lambda)\cdot\mathcal{L}_{task}
\end{equation}

By introducing task-based loss, the network may sacrifice some predictive accuracy, but the framework gains several advantages when applied to real-world power grids:

\vspace{0.2cm}

\noindent\textbf{Statistical Efficiency. }During training, the convergence time and complexity for finding the optimal strategy are generally much lower than those required for building a comprehensive understanding. Similar to reinforcement learning, a robot does not need to develop human-like intelligence to accomplish its tasks. In this context, the network requires less data to achieve a balanced performance, allowing the framework to perform well even in data-scarce real-world scenarios.

\noindent\textbf{Robust Decision-making. }While the network's predictive accuracy may decrease, the introduction of conformal prediction enables us to quantify the coverage level of the predicted results. This allows the framework to make more robust decisions under worst-case scenarios, ensuring reliability in uncertain conditions.

\noindent\textbf{Explicit Constraint Expressions. }since the optimization model imposes constraints on the network's predictions, such as Constraint (\ref{7a}), large deviations in the network's predictions can lead to anomalies in the objective function of the optimization problem. These constraints inherently act as explicit hard constraints on the network predictions, which are typically challenging to impose directly within neural networks. This provides the network with a structured framework, enabling faster convergence during the initial training phase.

\vspace{0.2cm}

\section{Experiments}
\begin{figure}
    \centering
    \includegraphics[width=0.25\textwidth]{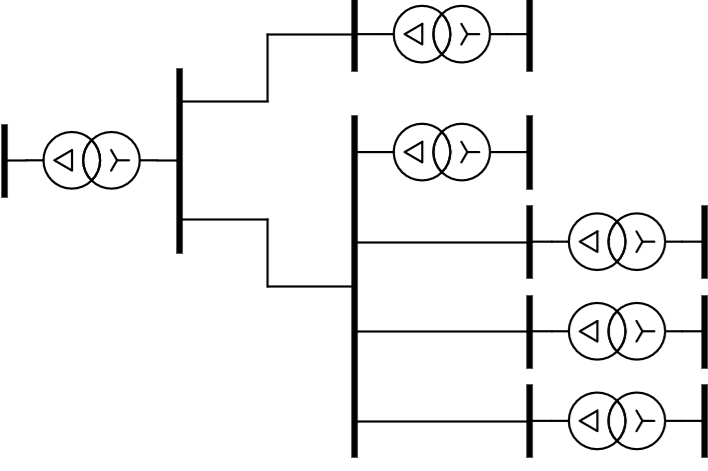}
    \caption{Simplified regional grid topology.}
    \label{fig2}
\end{figure}
\subsection{Dataset}
We evaluate our model using a dataset sourced from Southern California, which includes both weather data and grid data. For the grid component, we first extract the parameters of a real-world regional power grid. These parameters comprise bus indices, bus demands, inter-bus admittance matrices, real-time electricity prices, and slack bus voltage. Subsequently, to simplify the grid topology, we eliminate non-electrical components and model transformers between buses as equivalent injections. The final simplified topology is shown in Figure \ref{fig2}. This simplification ensures that the grid model retains essential electrical characteristics while reducing computational complexity. Additionally, we collect weather data, including wind speed, temperature, humidity, and solar radiation intensity, from the California Independent System Operator (CAISO) to complement the entire dataset.

\subsection{Simulation Results}
We first denote the traditional method as estimation-then-optimization (ETO), which is a widely used framework in decision-making under uncertainty. In this paradigm, the prediction task and the subsequent optimization task are decoupled. First, a predictive model estimates uncertain parameters, such as electricity prices or load demands, based on historical data. These estimates are then treated as fixed inputs for the optimization model, which determines the optimal decision strategy.

We first analyze the differences between the capabilities of our proposed method and the traditional ETO approach in terms of prediction accuracy and task performance. Subsequently, we present quantitative experiments and analyses to further evaluate our method.

\subsubsection{Prediction and Task Performance}
To evaluate the predictive capability and task-solving performance, the mean squared error (MSE) between the predicted values and ground-truth values, as well as the values of the optimization objective function, are utilized. In our method, predicted values are calculated as the mean of the upper and lower bounds of the box uncertainty. The experimental results are presented in Table \ref{tab1} and the $\mathcal{L}_{pred}$ and $\mathcal{L}_{task}$ losses are assigned different weights of $0.8$ and $0.2$.

\begin{table}[htb]
\centering
\begin{tabular}{cl|ccc}
\toprule
\#&\textbf{Method} & \textbf{Task Loss $\downarrow$} & \textbf{Prediction Loss $\downarrow$} & \textbf{Total Loss$ \downarrow$} \\
\midrule
1&\textbf{ETO}            & 682.91             & \textbf{332.62}          & 1015.53             \\
2&\textbf{Ours}$^*$       & \textbf{623.34}    & 341.27                   & \textbf{964.61}     \\
\bottomrule
\end{tabular}
\caption{Comparison of model performances.}
\label{tab1}
\end{table}

As shown in Table \ref{tab1}, our method demonstrates superior performance in terms of task loss and overall loss compared to the ETO approach, albeit with a slight reduction in predictive accuracy. This is because our method represents a trade-off between task-specific performance and predictive accuracy, placing greater emphasis on optimizing the task-specific strategy rather than achieving precise predictions.

\begin{figure}[htb]
    \centering
    \subfigure[Prediction effect at $\alpha=0.05$.]{%
        \label{fig:Contour_1}
        \includegraphics[width=0.40\textwidth]{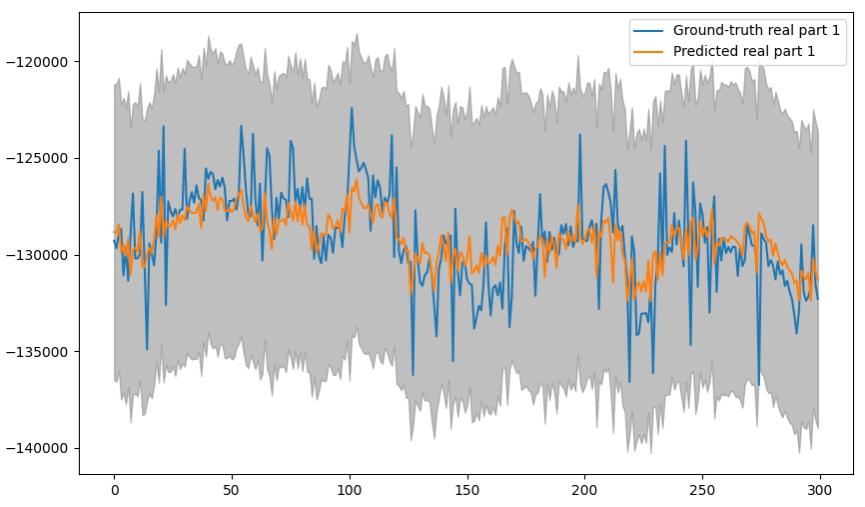}
    }\\
    \subfigure[Prediction effect at $\alpha=0.3$.]{%
        \label{fig:Contour_2}
        \includegraphics[width=0.40\textwidth]{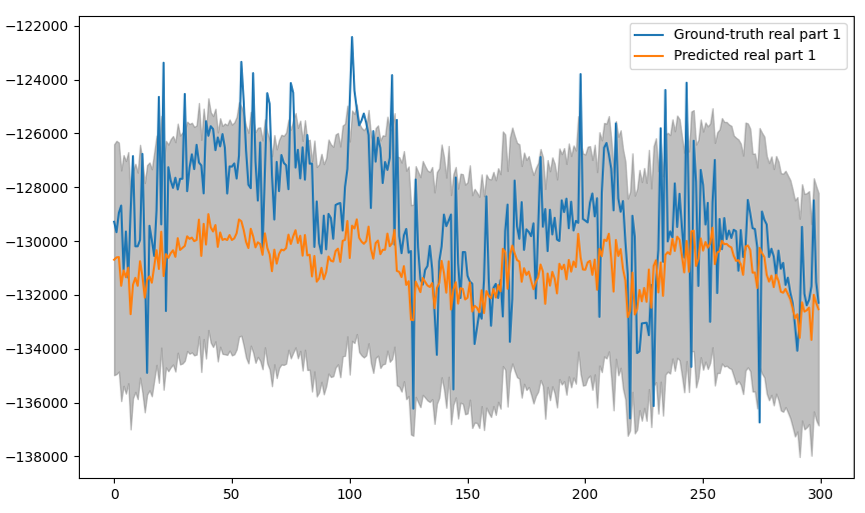}
    }%
    \caption{Predictive efficacy at varying confidence $\alpha$.}
    \label{fig3}
\end{figure}

\subsubsection{Loss Weights}
Table \ref{tab2} presents an analysis of the task performance under varying weights of $\mathcal{L}_{pred}$ and $\mathcal{L}_{task}$. In our experiments, the coverage level $\alpha$ is uniformly set to $0.1$. We observe that as the task weight increases from $0.1$ to $0.2$, both the total loss and task loss decrease. However, when the prediction loss weight falls below a certain threshold, the framework tends to converge quickly to a similar optimal decision, resulting in comparable task losses. In this scenario, the network struggles to provide accurate predictions, which ultimately undermines the practical significance of the derived optimal scheduling strategies. Therefore, it is critical to set appropriate weights for the loss components. We recommend setting the weight parameter $\lambda$ within the range of $0.7 - 0.9$, as this balance ensures that the framework can simultaneously maintain prediction accuracy and robust task performance. 

\begin{table}[htb]
\centering
\begin{tabular}{cc|ccc}
\toprule
&\textbf{Weight Ratio ($\lambda$)} & \textbf{Task} & \textbf{Prediction} & \textbf{Total} \\
\#&\textbf{$\mathcal{L}_{pred}:\mathcal{L}_{task}$} &      \textbf{Loss} &            \textbf{Loss} & \textbf{Loss} \\
\midrule
1&\textbf{0.9 : 0.1}                  & 663.27             & 334.80                   & 998.07        \\
2&\textbf{0.8 : 0.2}                  & 623.34             & 341.27                   & \textbf{964.61}\\
3&\textbf{0.7 : 0.3}                  & 617.71             & 349.53                   & 967.24        \\
4&\textbf{0.6 : 0.4}                  & 615.82             & 360.41                   & 976.23        \\
\bottomrule
\end{tabular}
\caption{Performance under different weight settings.}
\label{tab2}
\end{table}

\subsection{Coverage Level}
Figure \ref{fig3} illustrates the predictive performance of the model under varying coverage levels, visualized as the real part of the slack bus demand. As the confidence level $\alpha$ increases, the uncertainty interval provides better coverage of the ground-truth values. Consequently, higher confidence levels enhance the predictive capability and improve the rationality of strategies derived from the model. However, it is important to note that larger values of $\alpha$ result in better tasks and total loss. This is because, with a more relaxed feasible decision set, the optimal objective function value obtained is lower or equal to that of a stricter feasible set. However, such relaxed predictions may deviate from reality, leading to solutions that lack practical significance.

\section{Conclusion}
To provide actionable recommendations for real-world DER device deployments, we formulate a capacity expansion problem as a sequential decision-making model. To the best of our knowledge, this work represents the first attempt to incorporate a distribution network model into such a problem while simultaneously integrating neural networks for predicting uncertain parameters in an end-to-end manner. Subsequently, We propose a hybrid framework that seamlessly connects the optimization model with a neural network, enabling the entire framework to not only make accurate predictions but also adaptively select decisions based on the predictions. This adaptability ensures that the framework is more practical and deployable in the physical world. Finally, we validate the effectiveness of our approach using real-world data collected from Southern California. The results demonstrate the significance of our work for power grid scheduling research in real-world applications, highlighting its potential to improve decision-making in capacity expansion and DER device deployment.

\bibliography{conference}
\end{document}